\documentclass[pra,letterpaper,showpacs,twocolumn,amsfonts]{revtex4}

\usepackage{epic,eepic,latexsym,verbatim,revsymb}
\usepackage{theorem}

\usepackage{amsfonts}
\usepackage{amsmath}
\usepackage{amssymb}
\usepackage{dcolumn}
\usepackage{bm}
\usepackage{epsfig}
\usepackage{graphicx,graphics,rotating}
\usepackage{wrapfig}
\usepackage{dcolumn}
\usepackage{bbm,fancybox}
\usepackage{times,euscript,oldgerm}
\usepackage[german,english]{babel}
\usepackage{psfrag}
\usepackage{color}

\newtheorem{thm}{Theorem}

\newtheorem{lem}{Lemma}

\newenvironment{proof}[1][Proof]{\noindent\textbf{#1.} }{\ \rule{0.2em}{0.2em}}

\begin{document}

\title{Novel Schemes for Directly Measuring Entanglement of General States}
\author{Jianming Cai$^{1,2}$}
\email{Jian.Cai@uibk.ac.at}
\author{Wei Song$^{3,4}$}
\affiliation{$^{1}$Institut f{\"u}r Theoretische Physik, Universit{\"a}t Innsbruck,
Technikerstra{\ss }e 25, A-6020 Innsbruck, Austria \\
$^{2}$Institut f\"ur Quantenoptik und Quanteninformation der
\"Osterreichischen Akademie der Wissenschaften, Innsbruck, Austria \\
 $^{3}$Institute for Condensed Matter Physics, School of Physics and
Telecommunication Engineering, South China Normal University,
Guangzhou 510006, China\\
 $^{4}$Department of Modern Physics,
University of Science and Technology of China, Hefei 230026, China}
\date{\today}

\begin{abstract}
An intrinsic relation between maximally entangled states and
entanglement measures is revealed, which plays a role in
establishing connections for different entanglement quantifiers. We
exploit the basic idea and propose a framework to construct schemes
for directly measuring entanglement of general states. In
particular, we demonstrate that rank-$1$ local factorizable
projective measurements, which are achievable with only one copy of
entangled state involved at a time in a sequential way, are
sufficient to directly determine the concurrence of arbitrary
two-qubit entangled state.
\end{abstract}

\pacs{03.67.Lx, 03.67.Mn, 03.65.Ud, 03.65.Wj}
\maketitle

\textit{Introduction.} Quantum entanglement is one of the most
significant feature of quantum mechanics \cite{EPR35}, which has
attracted a lot of interest within the burgeoning field of quantum
information science and its intersection with many-body physics
\cite{Horodecki07,Amico07}. Entanglement measures play a central
role in the theory of entanglement. It is well known that
antilinearity from symmetries with time reversal operations is
intrinsically nonlocal, which leads to a natural routine to describe
and estimate entanglement \cite{Peres96,Horodecki96,Wootters98,Vidal02,Rungta01,Uhlmann00,Chen05}.
These entanglement measures based on nonphysical allowed
transformations are usually nonlinear functions of the density
matrix elements, and thus are difficult to determine directly in
experiments. It is worth to point out that there exists an
alternative experimental favorable class of entanglement
quantifiers, which are directly defined through the measurable
observables \cite{Mintert05,Meyer02}. The interesting problems are:
{\it How and why can these quantities from (anti)symmetric projections
serve as entanglement quantifiers? Is there any connection
between the above two different classes of entanglement measures?}

As far as determining entanglement is concerned, there are two
desirable features. The first is about the parametric efficiency
issue. It is inefficient and not necessary to obtain all the state
parameters as quantum state tomography \cite{James01}, in particular
when one considers high dimensional and multipartite quantum
systems. This concerns not only experimental determining
entanglement itself, but is related to the general theoretical
problem about extracting information efficiently from an unknown
quantum state with the least measurement cost \cite{Horodecki03}.
Second, in many realistic scenarios, entangled particles are shared
by two distant parties Alice and Bob, e.g. long distance quantum
communication. It will be valuable that Alice and Bob can measure
entanglement with only local operations on individual subsystems and
classical communications (LOCC).

The basic ideas of measuring these entanglement measures based on
nonphysical allowed transformations directly without state
reconstruction \cite{Horodecki03,Horodecki0206,Carteret03,Alves03,Bai06,Brun2004,
Yu08}, mainly rely on multiple copies of entangled state, i.e., a
number of entangled state need to be present at the same time. This could be difficult for
certain physical systems. The requisite experimental components
include structure physical approximation (SPA) and interferometer
circuit, the implementation of which with only LOCC is a great
challenge. One may wonder {\it Whether projective observables can  also help to
determine nonphysical allowed transformation based entanglement
measures?}

In this paper, we address the above problems by revealing an
intrinsic connection between maximally entangled states (MES) and
the definitions of nonphysical allowed transformation based
entanglement measures. The connection enables us to find that (anti)
symmetric projections can indeed extract the properties of density
matrices with nonphysical allowed transformations. This result opens
the possibility to establish relations between various kinds of
entanglement quantifiers
\cite{Wootters98,Vidal02,Rungta01,Uhlmann00,Chen05,Mintert05,Meyer02}.
The connection also allows us to propose a framework based on local
projections to design schemes for directly measuring the
entanglement quantifiers from nonphysical allowed transformations.
We explicitly demonstrate the benefit in determining entanglement of
general two-qubit states. The most remarkable feature is that only
one copy of entangled state need to be present at a time, which is
distinct from other schemes using multiple copies of entangled
state. As applications of our idea, we elucidate the physics
underlying the first noiseless quantum circuit for the
Peres-Horodecki separability criterion
\cite{Peres96,Horodecki96,Carteret05}, which was obtained in
\cite{Carteret05} through the mathematical analysis of polynomial
invariants \cite{Brun2004}. Moreover, one can easily construct a
circuit to directly measure the realignment properties of quantum
states \cite{Chen05}.

\textit{Connections between MES and entanglement measures.} One
useful tool in the entanglement theory is positive but not
completely positive map, with antilinear conjugation as the most
representative operation. Following the Peres-Horodecki separability
criterion \cite{Peres96,Horodecki96}, lots of entanglement detection
methods and measures based on the conjugation of density operator
have been established \cite{Wootters98,Vidal02,Rungta01,Uhlmann00}.
We propose to mathematically implement nonphysical allowed
transformations, in particular antilinear conjugation, with the
notation of MES. From an operational viewpoint, MES with appropriate
local unitary operations can be associated to (anti)symmetric
projective measurements. Thus, our result makes a connection between
antilinear conjugation \cite{Wootters98,Vidal02,Rungta01,Uhlmann00}
and (anti)symmetric projection based entanglement quantifiers
\cite{Mintert05,Meyer02}.

\begin{lem}
Given an $n$-partite operator $A$ on the Hilbert space $\mathcal{H}=\mathcal{%
H}_{1}\bigotimes \cdots \bigotimes \mathcal{H}_{n}$, with the dimension $dim(%
\mathcal{H}_{i})=d_{i}$, the maximally entangled state of $d_{i}\otimes
d_{i} $ bipartite system is denoted as $|S_{i}\rangle
=\sum_{s=0}^{d_{i}-1}|ss\rangle /\sqrt{d_{i}}$, then
\begin{equation}
(A\otimes \mathbf{I}_{\bar{1}\cdots \bar{n}})|\mathcal{S}\rangle=
(\mathbf{I}_{1\cdots n}\otimes A^{T})|\mathcal{S}\rangle \quad
\text{with} \quad
|\mathcal{S}\rangle=\bigotimes_{i=1}^{n}|S_{i}\rangle _{i\bar{i}}
\end{equation}
Moreover, we have
\begin{equation}
trA=(\prod_{i=1}^{n}d_{i})\langle \mathcal{S}|(\mathbf{I}_{1\cdots n}\otimes A)|\mathcal{S}\rangle
\end{equation}
\end{lem}

Eq.(1) is the generalization of the fact that qubit operator can travel through singlets, which
has been used to investigate the localizable entanglement properties of valence bond states \cite{Frank04}.
Here, from a different perspective, we view $A$ itself as a density matrix $\rho$ rather than an operator on
quantum states, lemma 1 thus indicates that with the notation of MES, we can mathematically implement the
(partial) transpose (conjugation) of arbitrary quantum states. Eq.(2) is another key point, which enables
us to extract the properties of transformed density operators through the projective measurements associated to
MES.

\textit{Remark 1:} Our idea is quite different from the SPA, in
which the transpose of quantum state is approximated by a completely
positive map \cite{Horodecki0206,Horodecki03}. It is worth to point
out that, in lemma 1, $A$ can be a density operator of arbitrary
dimensional multipartite quantum systems.

\begin{thm}
For a general quantum state $\rho $ on the Hilbert space $\mathcal{H}=%
\mathcal{H}_{1}\bigotimes \cdots \bigotimes \mathcal{H}_{n}$, we denote the antilinear transformation of $\rho $ as $ \tilde{\rho}%
_{u}=(U_{1}\otimes \cdots \otimes U_{n})\rho ^{\ast }(U_{1}^{\dagger
}\otimes \cdots \otimes U_{n}^{\dagger })$ and $|S_{u_{i}}\rangle
=(\mathbf{I}\otimes U_{i})|S_{i}\rangle $, it can be seen that
\begin{equation}
(\rho \otimes \rho )|\mathcal{S}_{u}\rangle=(\mathbf{I}_{1\cdots n}
\otimes \rho \tilde{\rho}_{u})|\mathcal{S}_{u}\rangle \quad
\text{with} \quad
|\mathcal{S}_{u}\rangle=\bigotimes_{i=1}^{n}|S_{u_{i}}\rangle
_{i\bar{i}}
\end{equation}%
This will lead to
\begin{equation}
tr(\rho \tilde{\rho}_{u})=(\prod_{i=1}^{n}d_{i})\cdot
tr[\bigotimes\limits_{i=1}^{n}\mathcal{P}_{u}^{(i\bar{i})}(\rho \otimes \rho )]
\end{equation}
where $U_{i}$ are local unitary operations, and
$\mathcal{P}_{u}^{(i\bar{i})} =|S_{u_{i}}\rangle _{i\bar{i}}\langle
S_{u_{i}}|$ are projections on two copies of the $i$-th subsystem.
\end{thm}

\textit{Remark 2:} Theorem 1 can help us to establish connections
between antilinearity and (anti)symmetric projections. The result is
quite general, e.g. $U_{i}$ can be arbitrary local unitary
operators, and it is applicable for high dimensional situations by
using appropriate $U_{i}$ or reducing the projections of high
dimensional bipartite systems into a sum of two-qubit projections
\cite{Mintert05}. It also provides an intuitive meaning of the
Wootters' concurrence, which can be linked with the success
probability of establishing MES via entanglement swapping following
the above theorem.

\textit{Novel schemes for measuring entanglement.} Besides the
theoretical interest, with the above connection we find that
projective observables can help to determine these nonphysical
transformation based entanglement quantifiers with much less
experimental efforts. We first demonstrate how to directly measure
the concurrence of general states
\cite{Wootters98,Rungta01,Uhlmann00}, and then explicitly illustrate
the physics underlying the first noiseless circuit for the
Peres-Horodecki separability criterion
\cite{Peres96,Horodecki96,Vidal02,Carteret05} following the present
idea. Finally, we construct a simple circuit for the realignment
separability criterion \cite{Chen05}.

\textit{I. Scheme for directly measuring the concurrence of general
states.} The concurrence family of entanglement measures are defined
through the eigenvalues $\lambda_{j}$ of $\rho\tilde{\rho}_{u}$ as
in theorem $1$. In order to determine these eigenvalues, quantum
state tomography need to obtain $(d_{1}\cdots d_{n})^{2}-1$
parameters, while direct strategy without state reconstruction only
need to measure the moments $m_{k}=\sum_{j}\lambda_{j}^{k}$, the
number of which is $d_{1}\cdots d_{n}$, and thus it is quadratically
efficient.

\begin{lem}
The moments of $\rho \tilde{\rho}_{u}$ can be obtained as follows
\begin{equation}
m_{k}=(d_{a}d_{b})^{k}tr[(\mathcal{P}^{(a)}\otimes \mathcal{P}^{(b)})\cdot V_{a_{2}\cdots
a_{2k}}\cdot V_{b_{2}\cdots b_{2k}}\bigotimes\limits_{i=1}^{2k}(\rho
)_{a_{i}b_{i}}]
\end{equation}%
$\mathcal{P}^{(s)}=\mathcal{P}_{u}^{(s_{1}s_{2})}\otimes \cdots \otimes
\mathcal{P}_{u}^{(s_{2k-1}s_{2k})}$, and $V_{s_{2}\cdots s_{2k}}$ are $k$-circle
permutations $(s=a,b)$.
\end{lem}

\begin{proof}
The $k$-cycle permutation
$V^{(k)}|\phi_{1}\rangle|\phi_{2}\rangle\cdots
|\phi_{k}\rangle=|\phi_{k}\rangle|\phi_{1}\rangle\cdots|\phi_{k-1}\rangle$
is the key element for spectrum measurement based on the property
$tr(V^{(k)}\bigotimes\limits_{i=1}^{k}A_{i})=tr(A_{k}\cdots A_{1})$
\cite{Ifmc}. As in lemma 1, using $2k$ copies of entangled state and
with the notation of MES, we can mathematically have $k$ copies of
$\rho\tilde{\rho}_{u}$. Thus, one can get Eq.(5) with the above two
facts.
\end{proof}

\textit{Remark 3:} For simplicity, we only give the formulations for
bipartite systems, lemma 2 however is valid for general multipartite
states. Our results provide a simple and general framework to design
schemes for directly measuring the concurrence family of entanglement
measures.

%Compared to constructing approximate completely positive
%maps, which could be complicated for high dimensional multipartite
%states,

If we use the similar interferometer circuit for spectrum measurement as
usual, $m_{k}$ can be obtained by controlled $k$-circle permutation and
experimental feasible antisymmetric projections, which means that half of
controlled-swap operations are saved compared with controlled $2k$-circle
permutation. Since $m_{k}$ are real, we do not have to measure the whole
interference pattern in order to obtain the visibility \cite{Horodecki0206,Horodecki03}.
Nevertheless, the implementation of interferometer circuit by LOCC is still complicated.
The experimental efforts can be reduced if no interferometer circuit is required.
We demonstrate the benefit of our framework in the case of general two-qubit states by
showing that only rank-$1$ local factorizable projective measurements are required, which then
leads to another interesting feature that we do not have to manipulate a number of entangled state at a time,
even the starting point of our scheme is also based on multiple copies of entangled state.

Consider a general two-qubit state $\rho$, its entanglement
can be quantified by the Wootters' concurrence as $C=\max\left\{ 0,\lambda_{1}-
\lambda_{2}-\lambda_{3}- \lambda_{4}\right\}$, where
$ \lambda_{i}$s are the square roots of the eigenvalues of $\rho\tilde{\rho}$ in the
decreasing order \cite{Wootters98}, with
$\tilde{\rho}=(\sigma_{y}\otimes\sigma_{y})\rho^{*}(
\sigma_{y}\otimes\sigma_{y})$. Before proceeding, we first introduce
some notations as $|\varphi _{0}\rangle =\otimes_{i=1}^{k}|S_{y}\rangle_{2i-1,2i}$, $%
|\varphi _{1}\rangle = V_{2,\cdots, 2k}|\varphi _{0}\rangle $ and $|\varphi
_{2}\rangle =-V_{2k-1,2k}|\varphi _{1}\rangle $, $|\varphi _{3}\rangle
=|\varphi _{1}\rangle-|\varphi _{2}\rangle $, where $V_{1, \cdots,
l}|\upsilon_{1}\rangle\cdots|\upsilon_{l}\rangle
=|\upsilon_{l}\rangle|\upsilon_{1}\rangle\cdots|\upsilon_{l-1}\rangle$.

\begin{thm}
Given $2k$ copies of general two-qubit state $\varrho_{2k}=\bigotimes%
\limits_{i=1}^{2k}(\rho)_{a_{i}b_{i}}$, the $k$-th moment $m_{k}$ can be
determined by the rank-$1$ local projective measurements as
\begin{eqnarray}
m_{1}&=&4 \langle \mathcal{P}^{(a)}_{0}\otimes \mathcal{P}^{(b)}_{0}\rangle \\ \nonumber
m_{k}&=&\frac{1}{4}m_{1}m_{k-1}+2^{2k}(\langle \mathcal{P}^{(a)}_{1}\otimes
\mathcal{P}^{(b)}_{1}\rangle-\langle \mathcal{P}^{(a)}_{2}\otimes \mathcal{P}^{(b)}_{2}\rangle)
\end{eqnarray}
where $k=2,3,4 $, and $\mathcal{P}_{0}=|\varphi _{0}\rangle \langle
\varphi_{0}|$, $\mathcal{P}_{l}=| \phi _{l}\rangle \langle
\phi_{l}|$ $(l=1,2)$ with $|\phi_{1}\rangle=(|\varphi _{0}\rangle
+|\varphi _{3}\rangle )/2$, $| \phi_{2}\rangle=(|\varphi _{0}\rangle
+i|\varphi _{3}\rangle )/2$. In particular, for $k=2$ the
expectation value $\langle \mathcal{P}^{(a)}_{1}\otimes
\mathcal{P}^{(b)}_{1}\rangle=m_{1}^{2}/16$.
\end{thm}

\begin{proof}
We denote $\mathcal{P}_{i,j,u,v}=\langle \varphi _{j}|\langle
\varphi _{i}|\varrho_{2k}|\varphi _{u}\rangle |\varphi _{v}\rangle
$. It can be seen that $V_{2k-1,2k}|\varphi _{0}\rangle =-|\varphi
_{0}\rangle $ and $V_{2k-1,2k}|\varphi _{3}\rangle =|\varphi
_{3}\rangle $, which leads to
$\mathcal{P}_{3,0,3,3}=-\mathcal{P}_{3,0,3,3}=0$. In the similar
way, one get $\mathcal{P}_{i,j,u,v}=0$ if the four indices are
either $0$ or $3$ and the number of $0$ is odd. After simple
calculations, we have $\mathcal{P}_{3,3,0,0}=4(\langle
\mathcal{P}^{(a)}_{1}\otimes \mathcal{P}^{(b)}_{1}\rangle-\langle
\mathcal{P}_{2}^{(a)}\otimes \mathcal{P}_{2}^{(b)}\rangle)$. Furthermore, we denote
$|\psi_{0}\rangle\equiv|\varphi_{1}\rangle+|\varphi_{2}\rangle$,
thus $\mathcal{P }=\langle \psi _{0}|\langle\psi
_{0}|\varrho_{2k}|\varphi_{0}\rangle
|\varphi_{0}\rangle=m_{1}m_{k-1}/2^{2k}$. Therefore, the $k$-th
moment is
\begin{equation}
m_{k}=2^{2k}\mathcal{P}_{1,1,0,0}=2^{2k}\cdot \frac{1}{4}(\mathcal{P}+\mathcal{P}_{3,3,0,0})
\end{equation}
We conclude that $m_{k}$ are measurable by only rank-$1$ local
projective observables as Eq.(6). For the case of $k=2$,
$|\phi_{1}\rangle=|\varphi _{1}\rangle$ which means that $\langle
\mathcal{P}^{(a)}_{1}\otimes
\mathcal{P}^{(b)}_{1}\rangle=m_{1}^{2}/16$.
\end{proof}

\textit{Remark 4:} Our scheme inherits the quadratic efficiency by
directly measuring four moments to determine the concurrence of
general two-qubit states. Only rank-$1$ local projective
measurements are required, which is expected to provide more flexibility in the
experiments.

\textit{II. Noiseless quantum circuit for the Peres separability
criteria.} The connection between MES and entanglement measures
plays its role not only in the concurrence family, but also in the
other scenarios. As an example, the first noiseless network to
measure the spectrum of a partial transposed density operator
\cite{Carteret05} from the structure of polynomial invariants, can
be recovered from a different perspective. Based on lemma
$1$, we can mathematically implement the partial transpose as
\begin{equation*}
(\mathbf{I}_{1}\otimes \rho _{\bar{1}2}\otimes
\mathbf{I}_{\bar{2}})|S\rangle _{1\bar{1}} |S\rangle _{2\bar{2}}
=\left[\mathbf{I}_{1}\otimes \mathbf{I}_{2}\otimes (\rho
^{T_{2}})_{\bar{1}\bar{2}}\right]|S\rangle _{1\bar{1}}|S\rangle
_{2\bar{2}}
\end{equation*}
We note that $V_{\bar{1}\bar{3}}|S\rangle _{1\bar{1}}|S\rangle _{3\bar{3}
}=V_{13}|S\rangle _{1\bar{1}}|S\rangle _{3\bar{3}}$, thus with $k$ copies
of entangled states we obtain the $k$-th moment of $\rho^{T_{2}}$ as
\begin{equation}
tr(\rho ^{T_{2}})^{k}= tr[V_{\bar{1}\cdots
\overline{2k-1}}V_{2\cdots 2k}^{\dagger }(\bigotimes_{l=1}^{k}\rho
_{\overline{2l-1}2l})]
\end{equation}
The righthand of Eq.(8) are exactly the circuits in
\cite{Carteret05}.

\textit{III. Realignment criterion for entanglement detection.}
Realignment of density operators, defined as $\mathcal{R}(\rho
)_{ij,kl}=\rho _{ik,jl}$, is another important operation to
establish separability criterions. Based on the trace norm of
$\mathcal{R}(\rho)$, Chen {\it et. al.} derived a low bound for the
concurrence of arbitrary dimensional $ d_{a}\otimes d_{b}$ bipartite
systems \cite{Chen05}. One can mathematically implement the
realignment as
\begin{eqnarray*}
V_{1} \cdot \rho _{12}\otimes \mathbf{I}_{\bar{1}\bar{2}}
|S\rangle _{1\bar{1}}|S\rangle_{2\bar{2}} &=&
 \mathbf{I}_{12}\otimes \left[\mathcal{R}(\rho)\right]_{\bar{1}\bar{2}}
|S\rangle _{1\bar{1}}|S\rangle _{2\bar{2}}  \\
V_{2}\cdot \rho _{12}\otimes \mathbf{I}_{\bar{1}\bar{2}}|S\rangle
_{1\bar{1}} |S\rangle _{2\bar{2}}&=&\mathbf{I}_{12}\otimes \left[
\mathcal{R}^{\dagger }(\rho )\right]_{\bar{1}\bar{2}}|S\rangle
_{1\bar{1}}|S\rangle _{2\bar{2}}
\end{eqnarray*}
where $V_{1}=V_{\bar{1}\bar{2}}V_{2\bar{2}}V_{12}$,
$V_{2}=V_{\bar{1}\bar{2} }V_{1\bar{1}}V_{12}$ and $|S\rangle
=\sum_{s=0}^{d-1}|ss\rangle $ with $d=\max \{d_{a},d_{b}\}$. Thus,
one can easily write the $k$-th moment of $\mathcal{R}(\rho
)\mathcal{R}^{\dagger }(\rho )$ as
\begin{equation*}
tr[\mathcal{R}(\rho )\mathcal{R}^{\dagger }(\rho )]^{k}=
tr[\bigotimes_{i=1}^{k}(V_{a_{2i-1}a_{2i}}V_{b_{2i-1}b_{2i-2}})\cdot
(\bigotimes_{i=1}^{2k}\rho _{a_{i}b_{i}})]
\end{equation*}
with $b_{0}=b_{2k}$, which enables us to construct a simple noiseless circuit.

\textit{Experimental implementation.} The main feature of our scheme
is taking advantage of the feasible (anti)symmetric projections to
access the properties of these nonphysical allowed transform based
entanglement quantifiers, which offers more flexibility in various
kinds of physical systems. We demonstrate in the following how
experimental efforts can be reduced in directly measuring the
concurrence of general two-qubit entangled states.

One can obtain $m_{1}$ through antisymmetric projective measurements
on two copies as Eq.(6). By noting that $|\phi _{2}\rangle
=U_{34}|\varphi _{1}\rangle $, where $U_{34}=i+(1-i)|S_{y}\rangle
\langle S_{y}|$, only one extra two-qubit gate for each party is
required to determine $m_{2}$. This can be achieved in certain
physical systems, e.g. optical lattice with engineered nearest
neighbor interactions \cite{Alves04}. To determine the higher
moments $m_{3}$ and $m_{4}$, more copies of entangled states are
required. One potential physical system is the ensembles of
multilevel quantum systems, in which $10$-$20$ qubits can be built
in single trapped cloud of ground state atoms
\cite{Brion07,Tordrup08}. The single element in our scheme, i.e.
rank-$1$ local projective measurement, is more favorable in such
physical systems than the conventional quantum circuits.

\begin{figure}[tbh]
\epsfig{file=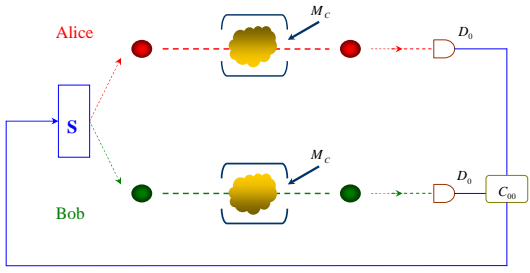,width=7cm} \caption{(Color online)
Implementation of rank-$1$ local projective measurements with
entangled pairs generated in a sequential way. The preceding
measurement results determine whether to generate the remain copies
or to restart the iteration. }
\end{figure}

Furthermore, the requirement for the schemes on multiple copies of
entangled state
\cite{Horodecki03,Horodecki0206,Carteret03,Alves03,Bai06,Brun2004,
Yu08} that a number of entangled state (up to $8$ copies ) have to
be present at the same time, becomes unnecessary in our scheme by
utilizing the intriguing matrix product state (MPS) formalism
\cite{AKLT,Frank06}. Every pure state $| \psi \rangle $ has an MPS
representation, and thus can be generated in a sequential way, i.e.
$V_{[2k]}\cdots V_{[1]}| \varphi _{L}\rangle _{\mathcal{C}}|0\cdots
0\rangle _{1\cdots 2k}=| \varphi_{R}\rangle _{\mathcal{C}}| \psi
\rangle_{1\cdots 2k}$, where $| \varphi_{L}\rangle_{\mathcal{C}}$
and $| \varphi _{R}\rangle_{\mathcal{C}}$ are the initial and final
state of an auxiliary system, e.g. cavity mode or atoms
\cite{Schon0507}. $V_{[i]}$ is a unitary interaction between qubit
$i$ and $\mathcal{C}$. By reversing the above procedure, we can
obtain the rank-$1$ local projective observable $ \langle \psi |
\langle \psi |\varrho_{2k}| \psi \rangle | \psi\rangle $ of Eq.(6)
in a similar sequential way as in Fig.1. First, the auxiliary system
is prepared in $| \varphi_{R}\rangle $, entangled pairs are
generated one by one, pass through and interact with $\mathcal{C}$,
then measured along the $\hat{z}$ basis. Only if the results of all
steps are $00$, we need to measure the auxiliary system with
$M_{\mathcal{C}}=|\varphi _{L}\rangle _{\mathcal{C}}\langle
\varphi_{L}\vert $. Otherwise, if the result of any step is not
$00$, we restart the iteration and do not need to generate all the
$2k$ copies. Our rough estimation shows that in comparison with
quantum state tomography, for each observable, the involved qubits
are a little more ($5/4$ and $4/3$ {\it vs.} $1$); however, the
total number of entangled states need to be generated is even less
($95/12$ {\it vs.} $9$).

\textit{Remark 5:} All current schemes for directly measuring
entanglement also raise an interesting problem: \textit{How does
entanglement play its role in reducing the measurement cost in
extracting information from an unknown quantum state}?

%This question is not only important in the entanglement theory, but
%also provides another perspective to understand the fundamental
%problem: \textit{what kind of entanglement features are required to
%obtain the speed-up of quantum information processing?}

\textit{Conclusions.} Maximally entangled state retains its
fundamental role in the entanglement theory, which provides an
approach to investigate the connections between different
entanglement quantifiers. With the notation of maximally entangled
states, one can mathematically implemented nonphysical allowed
transformations of quantum states. This enables us to design novel
schemes for directly measuring various kinds of entanglement
quantifiers. The benefit is explicitly demonstrated for general
two-qubit states, in which only rank-$1$ local projective
measurements are required. It is parametrically efficient without
increasing the requirement for state generation over quantum state
tomography.

\textit{Acknowledgments.} J.-M. C thanks H.-J. Briegel, O. G\"{u}hne
and B. Kraus for the helpful discussions, and is grateful for
support from the FWF through the Lise Meitner Program, and QICS (EU).

\end{document}